\documentclass[journal]{IEEEtran}

\usepackage{latexsym}
\usepackage{amssymb}
\usepackage{epsfig}
\usepackage{amsthm} 
\usepackage{subfigure}
\usepackage{graphics,color}
\usepackage{graphicx}
\usepackage{amsmath}
\usepackage{algorithm}
\usepackage{algorithmic}
\usepackage{cite}
\usepackage{comment}
\usepackage{url}
\usepackage{xspace}
\usepackage{psfrag}
\usepackage{booktabs}

\def\Pe{{\bar{ \text {Pe}}}}

\newcommand{\indc}{\mathbf{1}}

 \DeclareMathOperator*{\argmax}{arg\,max}
 \DeclareMathOperator*{\argmin}{arg\,min}
\newtheorem{theorem}{Theorem}
\newtheorem{lemma}{Lemma}
\newtheorem{proposition}{Proposition}
\newtheorem{corollary}{Corollary}
\newtheorem{assumption}{Assumption}
\newtheorem{fact}{Fact}
\theoremstyle{definition}
\newtheorem*{definition} {Definition}
\newtheorem{remarks}{Remark}

\newfont{\boldlarge}{msbm10 scaled 1100}
\newcommand{\Expt}{\mbox{\boldlarge E}}

\newcommand{\ignore}[1]{}

\renewcommand{\qed}{\nobreak \ifvmode \relax \else
      \ifdim\lastskip<1.5em \hskip-\lastskip
      \hskip1.5em plus0em minus0.5em \fi \nobreak
      \vrule height0.2em width0.5em depth0.4em\fi}
\IEEEoverridecommandlockouts
\begin{document}

\title{Sequentiality and Adaptivity Gains in\\ Active Hypothesis Testing}
\author{Mohammad Naghshvar and Tara Javidi
\thanks{This work was supported in part by the industrial sponsors of UCSD Center for Wireless Communication (CWC) and Center for Networked Systems (CNS), and NSF Grants CNS-0533035 and CCF-0729060.

The authors are with the Department of Electrical and Computer Engineering, University of California San Diego, La Jolla, CA 92093 USA. (e-mail: naghshvar@ucsd.edu; tjavidi@ucsd.edu).}
}

\maketitle

\begin{abstract}
Consider a decision maker who is responsible to collect observations so as to enhance his
information in a speedy manner about an underlying phenomena of interest.
The policies under which the decision maker selects sensing actions 
can be categorized based on the following two factors:
i) sequential vs. non-sequential; ii) adaptive vs. non-adaptive.
Non-sequential policies collect a fixed number of observation samples and make the final decision afterwards;
while under sequential policies, the sample size is not known initially and is determined by the observation outcomes.
Under adaptive policies, the decision maker relies on the previous collected samples to select the next sensing action;
while under non-adaptive policies, the actions are selected independent of the past observation outcomes.

In this paper, performance bounds are provided for the policies in each category.
Using these bounds, \emph{sequentiality gain} and \emph{adaptivity gain}, 
i.e., the gains of sequential and adaptive selection of actions are characterized.
\end{abstract}



\section{Introduction}

\thispagestyle{empty}

\newcommand{\tj}[1]{{\color{red}{(#1)}}}
\newcommand{\mn}[1]{{\color{blue}{#1}}}

This paper considers a generalization of the classical hypothesis testing problem. 
Suppose there are $M$ hypotheses among which only one is true. 
A Bayesian decision maker is responsible to enhance his information 
about the correct hypothesis in a speedy manner with a small number of samples
while accounting for the penalty of wrong declaration. 
In contrast to the classical \mbox{$M$-ary} hypothesis testing problem, at any given time,
our decision maker can choose one of $K$ available actions and hence, exert some control over the collected sample's ``information content.'' We refer to this generalization, originally tackled by Chernoff~\cite{Chernoff59}, as the \emph{active} hypothesis testing problem.
%
The special cases of active hypothesis testing naturally arise in a broad spectrum of applications
in cognition~\cite{Shenoy11}, communications~\cite{Burnashev76}, anomaly detection~\cite{ubli1}, image inspection~\cite{ubli2}, generalized search~\cite{Nowak11IT}, group testing~\cite{saligrama}, and sensor management~\cite{Hero11}.

The sample size and the sensing actions can be selected either 
based on the past observation outcomes (on-line)
or independent from them (off-line or open loop). According to this fact, the solutions are
divided into four categories based on the following two factors:
i) sequential vs. non-sequential; ii) adaptive vs. non-adaptive.
Non-sequential schemes collect a fixed number of observation samples and make the final decision afterwards; 
while under sequential ones, the sample size is not set in advance and instead is determined by the specific observations made. 
Under adaptive policies, the decision maker relies on the previous 
collected samples to select the next sensing action; while under 
non-adaptive policies, the actions are selected independent of the
past observation outcomes. 
A question of both theoretical and practical significance is the characterization of the benefits of making sequential 
and adaptive decisions relative to the non-sequential and non-adaptive solutions.

Due to the importance of the question, such gains have been characterized
for many special cases of the active hypothesis testing~\cite{ubli2, Nowak11, Iwen09}. 
For instance, in~\cite{Nowak11} and~\cite{Iwen09}, simple sequential and adaptive high dimensional reconstruction and sparse recovery are shown to significantly outperform the performance of the best non-sequential non-adaptive solutions. In contrast,  \cite{ubli2} identifies scenarios where the gain in practice is insignificant. In this paper, we consider the problem of active hypothesis testing in 
its full generality and provide upper and lower bounds on the expected cost of the optimal sensing selection strategies in sequential and non-sequential as well as adaptive and non-adaptive classes of policies.  Furthermore, the bounds are shown to be asymptotically tight (in terms of number of samples or equivalently in terms of reliability) and logarithmically increasing in the penalty of wrong declaration (or equivalently the error probability). 

As simple corollaries, we provide a full characterization of the sequentiality and adaptivity gains in the general active hypothesis testing framework. These findings generalize and extend those of~\cite{Nowak11} and~\cite{Iwen09} by showing
a logarithmic sequentiality gain in all cases and an additional logarithmic adaptivity gain in a large class of practically relevant cases.
Furthermore, the results prove, as a corollary, the conjecture given in~\cite{ubli2} on the insignificance of adaptivity gain when there exists a ``most informative'' sensing action which is independent of the Bayesian prior. 
Finally, we specialize our results in the active binary hypothesis testing case and state a simple 
 necessary and sufficient condition for a logarithmic adaptivity gain. 
 
  This work and analysis is closely related and complimentary to a growing body of literature on hypothesis testing \cite{Wald48,Armitage50,Lorden77,Chernoff59,Blahut74,Tuncel05,Hayashi09,PolyanskiyITA2011,NitinawaratArxiv,NitinawaratICASSP12}. We discuss the specific contributions and connections in Subsection~\ref{sec:survey}.

The remainder of this paper is organized as follows. 
In Section~\ref{sec:PS}, we formulate the problem and define various types of policies for selecting actions.
Sections~\ref{sec:main} and~\ref{sec:cons} provide the main results of the paper and 
discusses the advantage of sequential and adaptive selection of actions.
In Section~\ref{sec:binary}, active binary hypothesis testing is investigated as a special case
and a necessary and sufficient condition for a logarithmic adaptivity gain is provided.
Finally, we conclude the paper and discuss future work in Section~\ref{Discussion}.

\underline{Notations}:
A random variable is 
denoted by an upper case letter (e.g. $X$) and its realization is denoted by a lower 
case letter (e.g. $x$). 
For any set $\mathcal{S}$, $\left| \mathcal{S} \right|$ denotes the cardinality of $\mathcal{S}$.
For a set $\mathcal{A}$, let $\Lambda(\mathcal{A})$ denote the collection of all probability distributions on
elements of $\mathcal{A}$, i.e., $\Lambda(\mathcal{A}) = \{ \boldsymbol{\lambda} \in [0,1]^{|\mathcal{A}|}: \sum_{a \in \mathcal{A}} \lambda_a = 1 \}$.
%
%
The \emph{Kullback-Leibler (KL) divergence} between two probability density functions $q(\cdot)$ and $q'(\cdot)$ on space $\mathcal{Z}$ is
defined as $D(q||q')=\int_{\mathcal{Z}} q(z) \log\frac{q(z)}{q'(z)} dz$,
with the convention $0 \log \frac{a}{0}=0$ and $b \log \frac{b}{0}=\infty$ for $a,b\in [0,1]$ with $b\neq 0$.
The \emph{R\'enyi divergence of order $\alpha$}, $\alpha\in [0,1]$, between two probability density functions $q(\cdot)$ and $q'(\cdot)$ on space $\mathcal{Z}$ is denoted by $D_{\alpha}(q||q')$ where $D_{\alpha}(q||q')=\frac{-1}{1-\alpha} \log \int_{\mathcal{Z}} q^{\alpha}(z) {q'}^{1-\alpha}(z) dz$ for $\alpha\in [0,1)$ 
and $D_{\alpha}(q||q')=D(q||q')$ for $\alpha=1$. 
Finally, let $N(m,\sigma^2)$ denote a normal distribution with mean $m$ and variance $\sigma^2$. 

\section{Problem Setup}
\label{sec:PS}
In Subsection~\ref{sec:PF}, we formulate the problem of active hypothesis testing. 
Subsection~\ref{sec:Type} discusses different types of policies for selecting actions.
Subsection~\ref{sec:InfState} explains why active hypothesis testing is a partially observable Markov decision problem (POMDP) and provides the sufficient statistic for this problem.
Finally, in Subsection~\ref{sec:survey}, 
we state the main contributions of the paper and provide a summary of related works.

\subsection{Problem Formulation}
\label{sec:PF}

Here, we provide a precise formulation for the active \mbox{$M$-ary} hypothesis testing problem.

Let $\Omega=\{1,2,\ldots,M\}.$ 
Let $H_i$, $i \in \Omega$, denote $M$ hypotheses of interest among which only one holds true.  
 Let~$\theta$ be the random variable that takes the value $\theta=i$ on the event that $H_i$ is true for $i \in \Omega$. 
 We consider a Bayesian scenario with a given prior (belief) about~$\theta$,
 i.e., initially $P(\{\theta=i\})=\rho_i(0)>0$ for all $i \in \Omega$.
 $\mathcal{A}$~is the set of all sensing actions and is assumed to be finite with $|\mathcal{A}|=K < \infty$.
 $\mathcal{Z}$~is the \emph{observation space}. 
 For all $a \in \mathcal{A}$, the observation kernel $q^a_i(\cdot)$ (on $\mathcal{Z}$) 
 is the probability density function for observation~$Z$ when action $a$ has been taken and $H_i$ is true.
 We assume that observation kernels $\{q^a_i(\cdot)\}_{i,a}$ are known.
 Let~$L$ denote the penalty for a wrong declaration, i.e.,~the penalty of selecting $H_j$, $j \neq i$, when $H_i$ is true. 
 Let $\tau$ be the (stopping) time at which the decision maker retires.
 The objective is to find a stopping time $\tau$, 
 a sequence of sensing actions $A(0), A(1), \ldots, A({\tau-1})$, and a declaration rule 
 $d:~\mathcal{A}^{\tau} \times \mathcal{Z}^{\tau} \to \Omega$ that collectively minimize the expected total cost
	\begin{align}
	\label{Obj1}
	\Expt \left[ \tau + L {\bf{1}}_{\{d(A^{\tau},Z^{\tau}) \neq \theta\}} \right],
	\end{align}
 where the expectation is taken with respect to the initial belief as well as the distribution of observation sequence. 

Note that in the above problem, the cost of a test is stated in terms of minimizing the expected sample size plus the expected penalty of wrong declaration. We are interested in the characterization of this cost as a function of penalty $L$. 
 It is easy to show that under the optimal selection rule, the probability of error approaches zero as $L$ approaches infinity. 
 Furthermore, as shown in~\cite{HypJournal}, 
 the above problem is (asymptotically) equivalent to the problem of minimizing the (expected) number of samples 
 subject to a constraint $\epsilon=(L \log L)^{-1}$ on the expected probability of error.

\subsection{Types of Policies}
\label{sec:Type}

A \emph{policy} is a rule based on which stopping time $\tau$ and 
sensing actions $A(t)$, $t=0,1,\ldots,\tau-1$ 
are selected. We assume that sensing actions are selected 
according to randomized decision $\boldsymbol{\lambda} \in \Lambda(\mathcal{A})$
whose element $\lambda_a$ indicates the probability of selecting sensing action $a$ and 
in general might change with time or not. 
The sensing actions and the stopping time can be selected either 
based on the past observation outcomes
or independent from them.
According to this fact, policies are
divided into four categories based on the following two factors:
i) sequential vs. non-sequential; ii) adaptive vs. non-adaptive.
Non-sequential policies collect a fixed number of observation samples and make the final decision afterwards; 
while under sequential policies, the sample size is not known initially and is determined by the observation outcomes. 
More precisely, under non-sequential policies, $\tau=N$ for some $N \in \mathbb{N}$; 
while for sequential policies, $\tau$ is a random stopping time.   
Under adaptive policies, the decision maker relies on the previous 
collected samples to select the next sensing action; while under 
non-adaptive policies, the actions are selected independent of the
past observation outcomes.
%






\subsection{Information State as Sufficient Statistic}
\label{sec:InfState}

The problem of active \mbox{$M$-ary} hypothesis testing is a 
partially observable Markov decision problem (POMDP) 
where the state is static and observations are noisy. 
It is known that any POMDP is equivalent to an MDP 
with a compact yet uncountable state space, for which 
the belief of the decision maker about the 
underlying state becomes an information state~\cite{Kumar}.
In our setup, thus, the information state at time~$t$ is 
nothing but a belief vector specified by the conditional probability
of hypotheses $H_1, H_2, \ldots, H_{M}$ to be true 
given the initial belief and all the previous observations and actions.
Let $\boldsymbol{\rho}(t)$ denote the posterior belief after $t$ observations.
Accordingly, the information state space is defined as 
$\mathbb{P}(\Theta) = \big\{  {\boldsymbol{\rho}} \in [0,1]^M: \sum_{i=1}^{M} \rho_i = 1 \big\}$ 
where $\Theta$ is the \mbox{$\sigma$-algebra} generated by random variable $\theta$.
In one sensing step, the evolution of the belief vector follows Bayes' rule and 
the expected total cost (\ref{Obj1}) can be rewritten as
\begin{align}
\label{Obj2}
\Expt \left[ \tau \right]  + L \Pe,
\end{align}
where $\Pe = \Expt [1-\max_{j \in \Omega} \rho_j(\tau)]$ is the probability of wrong declaration
and the expectations are taken with respect to the distribution of observation sequence as well as
the prior distribution on $\theta$.

Let $V_{NN}(\boldsymbol{\rho})$, $V_{SN}(\boldsymbol{\rho})$, $V_{SA}(\boldsymbol{\rho})$, and $V_{NA}(\boldsymbol{\rho})$, 
denote the minimum expected total cost (\ref{Obj2}) for prior belief~$\boldsymbol{\rho}$ under
non-sequential non-adaptive, sequential non-adaptive, sequential adaptive, and non-sequential adaptive policies, respectively.

\subsection{Overview of the Results and Literature Survey}
\label{sec:survey}

Active hypothesis testing generalizes the passive (classical) hypothesis testing problem 
where the number of sensing actions is limited to one, both
in the fixed sample size (non-sequential) case \cite{Blahut74, Tuncel05, Haroutunian07} 
as well as the sequential one~\cite{Wald48, Armitage50, Lorden77}.  While the 
fixed sample size studies have primarily focused on the asymptotic analysis in form of 
identifying error exponents for various error types \cite{Blahut74, Tuncel05, Haroutunian07}, the 
study of sequential hypothesis testing has come in form of identifying the expected optimal 
sample size to achieve a given error probability.

The generalization to the active testing case was considered by Chernoff in~\cite{Chernoff59}  
in which a decision maker controls sensing actions to optimize the expected total cost~\eqref{Obj1} in a sequential (variable sample size)
setting.   
In particular, in~\cite{Chernoff59} and its extensions \cite{ISIT11, HypJournal, NitinawaratArxiv}, heuristic sequential adaptive randomized policies were proposed and were shown to be asymptotically optimal as $L \to \infty$ where the notion of asymptotic optimality~\cite{Chernoff59}
denotes the relative tightness of the performance upper bound associated with the
proposed policy and the lower bound associated with the optimal policy.\footnote{In~\cite{Chernoff59}, 
the objective was to minimize $c \mathbb{E}[\tau] + \Pe$ and 
the proposed policy was shown to be asymptotically optimal as $c\to 0$.
It is straightforward to show that for $L=\frac{1}{c}$, this problem coincides with the active hypothesis testing problem defined in this paper.
However, we have chosen $\mathbb{E}[\tau]+L \Pe$ as an objective function here
because of its Lagrangian relaxation interpretation of an information acquisition problem 
in which the objective is to minimize $\mathbb{E}[\tau]$ subject to $\Pe\le \epsilon$ 
where $\epsilon >0$ denotes the desired probability of error.} 

 The general active binary hypothesis testing problem was recently studied in \cite{Hayashi09,PolyanskiyITA2011} where full characterization of the error exponent corresponding to the class of adaptive and non-adaptive policies was provided.
In particular, the error exponent corresponding to these two classes was shown to be equal, hence establishing zero adaptivity gain among non-sequential policies. 
The generalization to $M>2$ was considered in \cite{NitinawaratArxiv}. 
Note that while \cite{NitinawaratArxiv} fully characterizes the error exponent corresponding to non-sequential non-adaptive policies; it provides only a partial characterization of (i.e., loose upper and lower bounds on) the error exponent corresponding to non-sequential adaptive policies.

 Table~\ref{tbl:HypLit} provides a visual summary of the literature on hypothesis testing, excluding the 
 authors' prior work, as discussed above. 

\begin{table}[htp]
\center
\caption{Hypothesis Testing Literature}
\label{tbl:HypLit}
\begin{tabular}{ccc}
  \toprule
  Type & $M=2$ & $M>2$ \\
  \midrule
  \vspace{0.05 in} 
  Sequential Passive ($K=1$)  &  \cite{Wald48} & \cite{Armitage50, Lorden77} \\
  \midrule
  \vspace{0.05 in} 
  Sequential Non-adaptive &  \cite{PolyanskiyITA2011} &  \\
  \midrule
  \vspace{0.05 in} 
  Sequential Adaptive  &  \cite{Chernoff59,PolyanskiyITA2011} & \cite{Chernoff59,NitinawaratArxiv} \\
  \midrule
  \vspace{0.05 in} 
  Non-sequential Passive ($K=1$) &  \cite{Blahut74} & \cite{Tuncel05} \\
  \midrule
  \vspace{0.05 in} 
  Non-sequential Non-adaptive  &  \cite{Hayashi09,PolyanskiyITA2011} & \cite{NitinawaratArxiv}  \\
  \midrule
  \vspace{0.05 in} 
  Non-sequential Adaptive &  \cite{Hayashi09,PolyanskiyITA2011} &  \cite{NitinawaratArxiv}  \\
  \bottomrule
\end{tabular}
\end{table}

We close our literature survey with an overview of the main contributions of this paper, which expands 
our previous works \cite{Allerton10, ISIT11, CISS2012, HypJournal}  and unifies various aspects of the prior work:
\begin{itemize}
	\item We provide asymptotically tight lower and upper bounds on $V_{NN}(\boldsymbol{\rho})$, $V_{SN}(\boldsymbol{\rho})$, and $V_{SA}(\boldsymbol{\rho})$ which hold uniformly for all prior $\boldsymbol{\rho}\in\mathbb{P}(\Theta)$.
	\begin{itemize}
		\item The asymptotic tight bounds on $V_{NN}(\boldsymbol{\rho})$ relies on the analysis of \cite{Blahut74, Tuncel05} 
		and the realization that in order to minimize the total cost, we have to decrease the error probabilities of various types
			with the same exponent among the worst pair of hypotheses. Since unlike the passive case studied in
			 \cite{Blahut74, Tuncel05}, the non-adaptive policies produce non-iid observation samples, 
			the final step is to characterize the relationship between the error exponent of a fixed block length 
			and one-step error exponent.
		\item The asymptotic tight bounds  on $V_{SN}(\boldsymbol{\rho})$ extend the results obtained by \cite{Lorden77} to 
			the Bayesian context while allowing for randomized non-adaptive policies. More specifically, the result of
			 \cite{Lorden77} is obtained via the law of large numbers and only holds if the observations are i.i.d. 
			 Since observations are not identical (although they are independent), different proof technique is required (note
			 that unlike the non-sequential case of extending the work of~\cite{Blahut74, Tuncel05}, the random nature  of 
			 sample size in the sequential case does not allow for a predetermined relationship between 
			 the error exponent of a fixed block and one-step error exponent).
		\item The asymptotic tight bounds on $V_{SA}(\boldsymbol{\rho})$ extend those obtained by Chernoff~\cite{Chernoff59}
		to the Bayesian context while relaxing the assumption on uniform discrimination of hypotheses or the need for the infinitely 
		often reliance on randomized action deployed in~\cite{NitinawaratArxiv} to ensure 
		sufficient discrimination among hypotheses.
	\end{itemize}
	\item In addition, we partially characterize a lower bound for $V_{NA}(\boldsymbol{\rho})$.  This is, in the
	Bayesian context, similar to the partial characterization of error exponent of~\cite{NitinawaratArxiv}. 
	\item As corollaries to the above performance bounds, we characterize the sequentiality gain and adaptivity gain in terms of $L$. In particular, it is shown that the sequentiality gain grows logarithmically as the penalty $L$ increases.  We also state a simple necessary and sufficient condition ensuring a logarithmic adaptivity gain in $L$ for the active binary hypothesis testing case.
	\item Furthermore, primarily as a sanity check, Section~\ref{sec:exp} contains the maximum achievable error exponents 
	$E_{NN}$, $E_{SN}$, and $E_{SA}$ in the Bayesian context. In particular, our result regarding $E_{NN}$ coincides with that of \cite{Hayashi09, PolyanskiyITA2011, NitinawaratArxiv}; while  the result regarding $E_{SA}$ coincides with that of \cite{Chernoff59, NitinawaratArxiv} in the Bayesian context. To the best of our knowledge, the result on $E_{SN}$ is new and has not been established before; while our upper bound on $E_{NA}$ is subsumed by the analysis in~\cite{NitinawaratArxiv}.
	
\end{itemize}

\section{Analytic Results}
\label{sec:main}

In this section, we provide the main results of the paper regarding the 
asymptotic characterization (in $L$) of 
$V_{NN}(\boldsymbol{\rho})$, $V_{SN}(\boldsymbol{\rho})$, 
$V_{SA}(\boldsymbol{\rho})$, and $V_{NA}(\boldsymbol{\rho})$. 

\subsection{Assumptions and Basic Definitions}
Throughout the paper, we have the following technical Assumptions.

\begin{assumption}
\label{KL0}
For any two hypotheses $i$ and $j$, $i \neq j$, there exists an action $a$, $a \in \mathcal{A}$, such that $D(q^a_i || q^a_j) > 0$.
\end{assumption}

\begin{assumption}
\label{Jump}
There exists $\xi < \infty$ such that
$$\max \limits_{i,j\in \Omega} \max \limits_{a \in \mathcal{A}} \sup \limits_{z \in \mathcal{Z}} \frac{q^a_i(z)}{q^a_j(z)} \le \xi.$$ 
\end{assumption}

Assumption~\ref{KL0} ensures the possibility of discrimination between any two hypotheses.
Assumption~\ref{Jump} implies that no two hypotheses are fully distinguishable using a single observation sample.

To continue with our analysis, we need the following definitions and notations.

\begin{definition}
For all $i \in \Omega$, $\boldsymbol{\lambda} \in \Lambda(\mathcal{A})$,
the \emph{optimized discrimination} of hypothesis~$i$ under randomized rule~$\boldsymbol{\lambda}$
is defined as
\begin{align*}
D^* (i,\boldsymbol{\lambda}):=\min \limits_{j \neq i} \max \limits_{\alpha \in [0,1]} (1-\alpha) \sum \limits_{a \in \mathcal{A}} \lambda_a  D_{\alpha} (q_i^a ||q_j^a).
\end{align*}
\end{definition}

\begin{definition}
For all $i \in \Omega$, $\boldsymbol{\lambda} \in \Lambda(\mathcal{A})$,
the \emph{reliability} function of hypothesis~$i$ with regard to randomized rule~$\boldsymbol{\lambda}$
is defined as
\begin{align*}
R(i,\boldsymbol{\lambda}):=\min \limits_{j \neq i} \sum \limits_{a \in \mathcal{A}} \lambda_a D(q^a_i||q^a_j),
\end{align*}
and the maximal randomized rule for hypothesis~$i$ is denoted by
\begin{align*}
\boldsymbol{\lambda}^*_i := \argmax_{\boldsymbol{\lambda} \in \Lambda(\mathcal{A})} R(i,\boldsymbol{\lambda}).
\end{align*}

For $\boldsymbol{\lambda} \in \Lambda(\mathcal{A})$, let $\bar{R}(\boldsymbol{\lambda})$ denote the harmonic mean of $\{R(i,\boldsymbol{\lambda})\}_{i \in \Omega}$, i.e.,
\begin{align*}
\bar{R}(\boldsymbol{\lambda}): = \frac{M}{\sum_{i=1}^M \frac{1}{R(i,\boldsymbol{\lambda})}},
\end{align*}
and let $\bar{R}^*$ denote the harmonic mean of $\{R(i,\boldsymbol{\lambda}^*_i)\}_{i \in \Omega}$, i.e.,
\begin{align*}
\bar{R}^*: = \frac{M}{\sum_{i=1}^M \frac{1}{R(i,\boldsymbol{\lambda}^*_i)}}.
\end{align*}
%
%
\end{definition}

These notions of discrimination and reliability, as we will see, are natural (and Bayesian) extensions of reliability in classical detection~\cite{Haroutunian07}
where reliability function for hypothesis~$i$ is related to type~$i$ error probability. The following fact 
enables a concrete relationship between these notions. 

\begin{fact}[Theorem~1 in~\cite{Shayevitz11ISIT}]
\label{RenyiKL}
For two probability density functions $q(\cdot)$ and $q'(\cdot)$ with the same support and for all $\alpha \in [0,1]$ we have 
\begin{align*}
(1-\alpha) D_{\alpha} (q ||q') \le \min \left\{ (1-\alpha) D (q ||q'), \alpha D(q'||q) \right\}.
\end{align*}
\end{fact}

\ignore{Let $V_{NN}$, $V_{SN}$, and $V_{SA}$, 
denote the minimum expected total cost respectively under
non-sequential non-adaptive, sequential non-adaptive, and sequential adaptive policies.
Next we state the main results of the paper, i.e., 
performance bounds on the minimum expected total cost 
under the policies mentioned above. 
The bounds presented here are based on a uniform prior belief, i.e., $\rho(0)=[1/M,\ldots,1/M]$
Generalizations of the bounds for arbitrary prior beliefs are proved in the appendix.}

\subsection{Main Theorems}

In this subsection, we provide upper and lower bounds on the 
minimum expected total cost (\ref{Obj1}) under different types of policies
defined in Subsection~\ref{sec:Type}.
These bounds will be used then in Section~\ref{sec:cons} to characterize the 
gains of sequential and adaptive selection of actions.

\begin{theorem}[Non-sequential non-adaptive policy]
\label{VNN}
Under Assumptions~\ref{KL0} and \ref{Jump},
\begin{align}
V_{NN} (\boldsymbol{\rho}) &\le 
\frac{\log L - \min \limits_{i,j \in \Omega} \log\frac{\rho_i}{\rho_j}}
{\hat{D}} + o(\log L),\\
\label{VNN01}
V_{NN} (\boldsymbol{\rho}) &\ge 
\frac{\log L - \max \limits_{i,j \in \Omega} \log\frac{\rho_i}{\rho_j}}
{\hat{D}} - o(\log L),
\end{align} 
where
\begin{align}
\label{D} 
\hat{D} : = \max \limits_{\boldsymbol{\lambda} \in \Lambda(\mathcal{A})} \min \limits_{i \in \Omega} D^* (i,\boldsymbol{\lambda}).
\end{align}
\end{theorem}

\begin{proof}
The detailed proof  is provided in Appendix~\ref{appendix:NN}. Here we provide an overview.

The proof of the lower bound relies on a generalization of Theorem~10 in~\cite{Blahut74}, while the upper bound
is achieved via a randomized, non-sequential, and non-adaptive policy which collects 
$
\hat{n}=\Big(\log L + \log (M-1) - \min \limits_{i,j \in \Omega} \log\frac{\rho_i}{\rho_j} + o(\log L) \Big)/\hat{D}
$
samples (deterministically) and selects sensing actions according to the randomization rule 
 $\hat{\boldsymbol{\lambda}} \in \Lambda(\mathcal{A})$ that 
achieves the maximum in (\ref{D}). 
\end{proof}

%


\begin{theorem}[Sequential non-adaptive policy]
\label{VSN}
Under Assumptions~\ref{KL0} and \ref{Jump}, 
\begin{align}
\label{VSN01-main}
V_{SN} (\boldsymbol{\rho}) &\le \min \limits_{\boldsymbol{\lambda} \in \Lambda(\mathcal{A})}
\sum_{i=1}^{M} \rho_i \frac{\log L - \min \limits_{k \neq i} \log\frac{\rho_i}{\rho_k}}{R(i,\boldsymbol{\lambda})} + o(\log L),\\
\label{VSN02-main}
V_{SN} (\boldsymbol{\rho}) &\ge \min \limits_{\boldsymbol{\lambda} \in \Lambda(\mathcal{A})}
\sum_{i=1}^{M} \rho_i \frac{\log L - \max \limits_{k \neq i} \log\frac{\rho_i}{\rho_k}}{R(i,\boldsymbol{\lambda})} - o(\log L).
\end{align}
\end{theorem}

\begin{proof}
The detailed proof  is provided in Appendix~\ref{appendix:SN}. Here we provide an overview.

Suppose $\hat{\boldsymbol{\lambda}} \in \Lambda(\mathcal{A})$ achieves the minimum in (\ref{VSN01-main}).
The upper bound \eqref{VSN01-main} is achieved by a policy that selects sensing actions according to $\hat{\boldsymbol{\lambda}}$
and stops sampling at $$\tau:=\min\{n:\max \limits_{i \in \Omega} \rho_i(n) \ge 1-L^{-1}\}.$$

From upper bound~\eqref{VSN01-main} we know that the total cost under the optimal policy is $O(\log L)$.
This implies that the error probability $\Pe$ of the optimal policy is $O(\frac{\log L}{L})$.
Hence, without loss of generality in our proof of the lower bound, 
we can restrict the set of  sequential and non-adaptive policies to those whose average probability of 
making an error is $O(\frac{\log L}{L})$. Conditioning on the true hypothesis and considering the dynamic 
of pairwise likelihoods,  we then compute the minimum expected number of samples necessary to achieve this target error probability.
\end{proof}

\begin{theorem}[Sequential adaptive policy]
\label{VSA}
Under Assumptions~\ref{KL0} and \ref{Jump}, 
\begin{align}
V_{SA} (\boldsymbol{\rho}) &\le 
\sum_{i=1}^{M} \rho_i \frac{\log L - \min \limits_{k \neq i} \log\frac{\rho_i}{\rho_k}}{R(i,\boldsymbol{\lambda}^*_i)} + o(\log L), \label{VSA-main-upper}\\
V_{SA} (\boldsymbol{\rho}) &\ge
\sum_{i=1}^{M} \rho_i \frac{\log L - \max \limits_{k \neq i} \log\frac{\rho_i}{\rho_k}}{R(i,\boldsymbol{\lambda}^*_i)} - o(\log L).
\end{align} 
\end{theorem}

\begin{proof}
The detailed proof is provided in Appendix~\ref{appendix:SA}. Here we provide an overview.

The proof of the lower bound relies on a generalization of Theorem~2 in~\cite{Chernoff59}. 
The upper bound is achieved via $\tilde{\pi}_1$, a heuristic two-phase policy introduced in~\cite{HypJournal} which
in its first phase, selects actions in a way that 
all pairs of hypotheses can be distinguished from each other;
while its second phase coincides with Chernoff's scheme~\cite{Chernoff59}
where only the pairs including the most likely hypothesis are considered.  
In~\cite{HypJournal}, the second phase of $\tilde{\pi}_1$ is shown to ensure its asymptotic optimality in $L$;
while its first phase in a very natural manner relaxes the technical assumption in~\cite{Chernoff59} 
where all actions are assumed to discriminate between all hypotheses pairs or the need for the infinitely 
often reliance on randomized action deployed in~\cite{NitinawaratArxiv} in order to ensure 
sufficient discrimination among hypotheses.
\end{proof}



We close this section by a note on the class of non-sequential adaptive policies 
even though they seem rather unnatural to us (It is more reasonable to control 
the sample size using the observation outcomes if they are already being used 
to select sensing actions). Next proposition provides a lower bound on 
the minimum expected total cost under non-sequential adaptive policies, denoted by $V_{NA}$.
\begin{proposition}[Non-sequential adaptive policy]
\label{VNA}
Under Assumptions~\ref{KL0} and \ref{Jump},
\begin{align}
\label{P4Bnd}
V_{NA} (\boldsymbol{\rho}) &\ge
\frac{\log L - \max \limits_{k \neq i} \log\frac{\rho_i}{\rho_k}}
{\min \limits_{i \in \Omega} \max \limits_{\boldsymbol{\lambda} \in \Lambda(\mathcal{A})} R(i,\boldsymbol{\lambda})} - o(\log L).
\end{align} 
\end{proposition}

Next we state and discuss the consequences of the bounds proposed above. 
In Subsection~\ref{sec:seq}, we focus on the advantages of causally selecting the retire/declaration time
as well as the adaptive selecting of sensing actions.
In Subsection~\ref{sec:exp}, we derive the error exponent corresponding to different types of policies.

\section{Consequences of the Bounds}
\label{sec:cons}

In this section, we first specialize and simplify the results provided in Section~\ref{sec:main} for uniform prior. In particular, 
assume that the hypotheses, initially, are equally likely, i.e., $\rho_i(0)=\frac{1}{M}$ for all $i\in\Omega$.  
Let $\mathbb{E}[\tau^*_{NN}]$, $\mathbb{E}[\tau^*_{SN}]$, and $\mathbb{E}[\tau^*_{SA}]$, 
denote the minimum expected number of samples 
under non-sequential non-adaptive, sequential non-adaptive, and sequential adaptive policies; while $\Pe_{NN}$, $\Pe_{SN}$, and $\Pe_{SA}$ represent average probability of making a wrong declaration.

From Fact~\ref{RenyiKL}, we know that
\begin{align}
\label{VNN02}
\hat{D} 
&\le 0.5 \max \limits_{\boldsymbol{\lambda} \in \Lambda(\mathcal{A})} \min \limits_{i \in \Omega}\min \limits_{j \neq i} 
 \sum_{a \in \mathcal{A}} \lambda_a D (q_i^a ||q_j^a).
\end{align}
Theorem~\ref{VNN} together with \eqref{VNN02} implies that:
\ignore{ \begin{proposition} 
\label{VNNprop}
Under Assumptions~\ref{KL0} and \ref{Jump},
\begin{align}
V_{NN} + L \Pe &\ge
\frac{2 \log L}{\max \limits_{\boldsymbol{\lambda} \in \Lambda(\mathcal{A})} \min \limits_{i \in \Omega} R(i,\boldsymbol{\lambda})} - o(\log L).
\end{align} 
\end{proposition}}

\begin{corollary}[Non-sequential non-adaptive policy]
\label{VNN_cor}
Under Assumptions~\ref{KL0} and \ref{Jump},
\begin{align}
\mathbb{E}[\tau^*_{NN}] + L \Pe_{NN} & = 
\frac{\log L}{\hat{D}} \pm o(\log L) \nonumber \\
&\ge
\frac{2 \log L}{\max \limits_{\boldsymbol{\lambda} \in \Lambda(\mathcal{A})} \min \limits_{i \in \Omega} R(i,\boldsymbol{\lambda})} - o(\log L).
\end{align} 
\end{corollary}

\begin{corollary}[Sequential non-adaptive policy]
\label{VSN_cor}
Under Assumptions~\ref{KL0} and \ref{Jump},
\begin{align}
\label{P2Bnd}
\mathbb{E}[\tau^*_{SN}] + L \Pe_{SN} &=
\frac{\log L}{\max \limits_{\boldsymbol{\lambda} \in \Lambda(\mathcal{A})} \bar{R}(\boldsymbol{\lambda})} \pm o(\log L).
\end{align} 
\end{corollary}

\begin{corollary}[Sequential adaptive policy]
\label{VSA_cor}
Under Assumptions~\ref{KL0} and \ref{Jump},
\begin{align}
\label{P3Bnd}
\mathbb{E}[\tau^*_{SA}] + L \Pe_{SA} &=
\frac{\log L}{\bar{R}^*} \pm o(\log L).
\end{align} 
\end{corollary}

\begin{remarks}
Note that the simple two phase structure of the policy which achieves the upper bound in (\ref{VSA-main-upper}) implies that the 
adaptivity gain can be obtained via coarse level adaptation. 
\end{remarks}

From the results above, it is evident that the minimum expected total cost 
under all classes of policies grows logarithmically in $L$.
However, the coefficient of the $\log L$ term is not the same in general
and we have
\begin{align}
\label{DenCmp}
\bar{R}^* &\ge {\max \limits_{\boldsymbol{\lambda} \in \Lambda(\mathcal{A})} \bar{R}(\boldsymbol{\lambda})} \ge 
	{\max \limits_{\boldsymbol{\lambda} \in \Lambda(\mathcal{A})} \min \limits_{i\in\Omega} R(i,\boldsymbol{\lambda})} \ge \hat{D}. 
\end{align}

\subsection{Sequentiality and Adaptivity Gains}
\label{sec:seq}

In this subsection, we discuss the advantage of causally selecting the retire/declaration time, i.e., $\tau$ as well 
as the sensing actions. Let $V_{NN}$, $V_{SN}$, and $V_{SA}$, respectively, 
denote the minimum expected total cost  under
non-sequential non-adaptive, sequential non-adaptive, and sequential adaptive policies under uniform prior, i.e., 
$V_x:= V_x([\frac{1}{M}, \frac{1}{M}, \ldots, \frac{1}{M}])$ where $x$ denotes the class of policies $NN$, $SN$, and $SA$.

First, we show that the performance gap between the sequential and non-sequential policy,
$V_{NN}-V_{SN}$, 
grows logarithmically as the penalty $L$ increases. 
We refer to this performance gap as the \emph{sequentiality gain}.

\begin{corollary}
\label{cor:seq}
Under Assumptions~\ref{KL0} and \ref{Jump}, the sequentiality gain is characterized as
\begin{align*}
\lefteqn{V_{NN}-V_{SN}}\\
 &\ge \log L \left(\frac{2}{\max \limits_{\boldsymbol{\lambda} \in \Lambda(\mathcal{A})}
 \min \limits_{i \in \Omega} R(i,\boldsymbol{\lambda})}
-\frac{1}{\max \limits_{\boldsymbol{\lambda} \in \Lambda(\mathcal{A})} \bar{R}(\boldsymbol{\lambda})}\right) - o(\log L).
\end{align*}
\end{corollary}

\begin{remarks}
\label{rem:seq}
The sequentiality gain grows logarithmically with~$L$ and from~(\ref{DenCmp}),
\begin{align*}
V_{NN}-V_{SN} &\ge \frac{\log L}{\max \limits_{\boldsymbol{\lambda} \in \Lambda(\mathcal{A})} \bar{R}(\boldsymbol{\lambda})} - o(\log L).
\end{align*}
\end{remarks}



\ignore{ 
\subsection{Adaptivity Gain}
\label{sec:adp}}

Next, the advantage of adaptively selecting the sensing actions is discussed. 
In particular, it is shown that the performance gap between the adaptive and non-adaptive policy,
$V_{SN}-V_{SA}$, grows logarithmically as the penalty $L$ increases. 
We refer to this performance gap as the \emph{adaptivity gain}.

\begin{corollary}
\label{cor:adp}
Under Assumptions~\ref{KL0} and \ref{Jump}, the adaptivity gain is characterized as
\begin{align*}
V_{SN}-V_{SA} = \log L \left(\frac{1}{\max \limits_{\boldsymbol{\lambda} \in \Lambda(\mathcal{A})} \bar{R}(\boldsymbol{\lambda})} - \frac{1}{\bar{R}^*} \right) \pm o(\log L).
\end{align*}
\end{corollary}

\begin{remarks}
\label{rem:adp}
Unless there exists a $\tilde{\boldsymbol{\lambda}} \in \Lambda(\mathcal{A})$ such that, 
\begin{align*}
R(i,\tilde{\boldsymbol{\lambda}}) = R(i,\boldsymbol{\lambda}^*_i) \text{  for all  } i \in \Omega,
\end{align*}
the adaptivity gain grows logarithmically with $L$.
\end{remarks}



A sufficient condition under which
there is no adaptivity gain is that of stochastic dominance/degradation \cite{Blackwell53}, i.e., if there exists a 
\emph{stochastic transformation} $W$ from $\mathcal{Z}$ to $\mathcal{Z}$ and\footnote{Function $W:\mathcal{Z} \times \mathcal{Z} \to \mathbb{R_+}$ is called a \emph{stochastic transformation} from $\mathcal{Z}$ to $\mathcal{Z}$ if it satisfies $\int_{\mathcal{Z}} W(y;z) dz = 1$.}
a sensing action $a^*$ such that for all 
other sensing actions $a\in \mathcal{A}$,
\begin{align}
\label{BlackwellCondition}
q^a_i(z)=\int q^{a*}_i(y) W(y;z) dy, \ \ \forall i \in \Omega.
\end{align}
As shown by Sakaguchi \cite{Sakaguchi64}, \eqref{BlackwellCondition} implies that 
$$D(q^a_i || q^a_j) \le D(q^{a^*}_i || q^{a^*}_j), \ \ \forall a\in\mathcal{A}, \ \forall i,j \in\Omega,$$ hence, ensuring zero adaptivity gain when observations obtained by 
all actions are stochastically degraded version of the observation under sensing action $a^*$.
This formalizes the notion of informativeness and confirms the conjecture provided in~\cite{ubli2}. 
 
\ignore{
\begin{definition}[Blackwell Ordering~\cite{Blackwell53}]
Given two conditional probability densities $q^a$ and $q^b$ from $\Omega$ to $\mathcal{Z}$, we say that $q^b$ is \emph{less informative than} $q^a$ ($q^b \le_B q^a$) if there exists a \emph{stochastic transformation} $W$ from $\mathcal{Z}$ to $\mathcal{Z}$ such that\footnote{Function $W:\mathcal{Z} \times \mathcal{Z} \to \mathbb{R_+}$ is called a \emph{stochastic transformation} from $\mathcal{Z}$ to $\mathcal{Z}$ if it satisfies $\int_{\mathcal{Z}} W(y;z) dz = 1$.}
\begin{align}
\label{BlackwellCondition}
q^b_i(z)=\int q^a_i(y) W(y;z) dy \ \ \text{for all $i \in \Omega$}.
\end{align}
\end{definition}

\begin{fact}[see~\cite{Sakaguchi64}]
Stochastic 
\end{fact}

\begin{corollary}
If there exists a sensing action $a^*$ satisfying $q^{a} \le_B q^{a^*}$ for all $a \in \mathcal{A}$,
then 
there is no adaptivity gain. 
\end{corollary} 
}

\subsection{Reliability and Error Exponent}
\label{sec:exp}

Let $\mathbb{E}^{\pi}[\tau]$ denote the expected stopping time 
(or equivalently the expected number of collected samples) under policy $\pi$.
Policy $\pi$ is said to achieve error exponent $E>0$ if 
\begin{align}
\lim_{t\to\infty} \frac{-1}{t} \log \text{Pe}^{\pi}(t,M) = E,
\end{align} 
where $\text{Pe}^{\pi}(t,M)$ is the smallest probability of error that policy $\pi$ can guarantee when looking for the true hypothesis among $M$ hypotheses with $\mathbb{E}^{\pi}[\tau] \le t$ (Note that for non-sequential policies, $\tau$ is deterministic).

Next we use the bounds obtained in Section~\ref{sec:main} 
to characterize 
the maximum achievable error exponent for different type of policies. 
Let $E_{NN}$, $E_{SN}$, $E_{SA}$, and $E_{NA}$
denote the maximum achievable error exponent under
non-sequential non-adaptive, sequential non-adaptive, sequential adaptive, and non-sequential adaptive policies.

\begin{corollary}
\label{MaxE} Under Assumptions~\ref{KL0} and \ref{Jump}, we have
\begin{align*}
E_{NN} &= \hat{D}  \\
E_{SN} & = \max \limits_{\boldsymbol{\lambda} \in \Lambda(\mathcal{A})} \bar{R}(\boldsymbol{\lambda}),\\
E_{SA} & = \bar{R}^*.
\end{align*}
\end{corollary}

\begin{remarks}
The above characterizations of maximum achievable error exponent are nothing but the Bayesian and $M$-ary version of the results 
in the literature (see Table~\ref{tbl:HypLit}). In fact as discussed in Subsection~\ref{sec:survey}, these results provide a sanity check viz a viz the 
prior work:  $E_{NN}$ coincides with that of \cite{Hayashi09, PolyanskiyITA2011, NitinawaratArxiv}; while 
$E_{SA}$ coincides with that of \cite{Chernoff59, NitinawaratArxiv}. To the best of our knowledge, the result on 
$E_{SN}$ is new and has not been established before.
\end{remarks}

\begin{remarks}
The above corollary provides alternative means to underline and characterize the sequentiality and adaptivity gains. 
In particular, sequentiality always results in an improvement in the maximum achievable error exponent since $E_{NN}  \le 0.5 \max \limits_{\boldsymbol{\lambda} \in \Lambda(\mathcal{A})} \min \limits_{i \in \Omega} R(i,\boldsymbol{\lambda}) < E_{SN}$. 
In contrast, adaptive selection of actions results in an improvement in the maximum achievable error exponent only if $\max \limits_{\boldsymbol{\lambda} \in \Lambda(\mathcal{A})} \bar{R}(\boldsymbol{\lambda}) \neq \bar{R}^*$.
\end{remarks}

We can also find an upper bound on the maximum achievable error exponent of any non-sequential yet adaptive policy (tight lower bounds are necessary for full characterization, however). 
\begin{corollary}
\label{MaxE-NA} Under Assumptions~\ref{KL0} and \ref{Jump}, we have
\begin{align*}
E_{NA} & \le \min \limits_{i \in \Omega} \max \limits_{\boldsymbol{\lambda} \in \Lambda(\mathcal{A})} R(i,\boldsymbol{\lambda}).
\end{align*}
\end{corollary}

\begin{remarks}
Our upper bound on $E_{NA}$ is subsumed by~\cite[Theorem~3]{NitinawaratArxiv}.
\end{remarks}

\section{Special Case: Binary Hypothesis Testing}
\label{sec:binary}

In this section, we consider active binary hypothesis testing ($M=2$)
as a special case.
%

\subsection{Analytical Results}

The performance bounds provided in Section~\ref{sec:main} 
are simplified by substituting the following equations into 
the denominators of the bounds.


\begin{align*}
R(1,\boldsymbol{\lambda}) &=
\sum \limits_{a \in \mathcal{A}} \lambda_{a} D(q^{a}_1||q^{a}_2),
\ R(2,\boldsymbol{\lambda}) =
\sum \limits_{a \in \mathcal{A}} \lambda_{a} D(q^{a}_2||q^{a}_1),\\
R(1,\boldsymbol{\lambda}^*_1) &=
\max \limits_{a \in \mathcal{A}} D(q^{a}_1||q^{a}_2),
\ \ \ R(2,\boldsymbol{\lambda}^*_2) =
\max \limits_{a \in \mathcal{A}} D(q^{a}_2||q^{a}_1),\\
\bar{R}(\boldsymbol{\lambda}) &=
\Bigg(\frac{0.5}{\sum \limits_{a \in \mathcal{A}} \lambda_{a} D(q^{a}_1||q^{a}_2)} +
\frac{0.5}{\sum \limits_{a \in \mathcal{A}} \lambda_{a} D(q^{a}_2||q^{a}_1)}\Bigg)^{-1},\\
\bar{R}^* &=
\left(\frac{0.5}{\max \limits_{a} D(q^{a}_1||q^{a}_2)} +
\frac{0.5}{\max \limits_{a} D(q^{a}_2||q^{a}_1)}\right)^{-1}.
\end{align*} 
 
Next we state a simple necessary and sufficient condition 
for a logarithmic adaptivity gain
in the active binary hypothesis testing case.

\begin{corollary} \label{binary-no-gain}
In the active binary hypothesis testing case,
the adaptivity gain grows logarithmically in $L$
if and only if
$$\argmax \limits_{a \in \mathcal{A}} D(q^a_1||q^a_2) \neq \argmax \limits_{a \in \mathcal{A}} D(q^a_2||q^a_1).$$
\end{corollary}

The problem of passive binary hypothesis testing ($K=1$, $M=2$) 
with fixed-length (non-sequential) as well as variable-length (sequential) sample size has been studied 
by~\cite{Blahut74, Csiszar04, Haroutunian07, NitinawaratArxiv}. Our sequentiality gain, in this case, is the manifestation 
of the fact that ``sequential tests are superior in ensuring that both error probabilities decreasing at the best possible exponential rates''~\cite{Csiszar04}. 


Recently, the authors in~\cite{Hayashi09} and~\cite{PolyanskiyITA2011} have studied the problem of active binary hypothesis testing 
for fixed-length and variable-length sample size respectively.
Our work complements the findings in~\cite{PolyanskiyITA2011} by providing
an asymptotic optimal solution in a total cost (and Bayesian) sense as well as establishing 
a non-zero sequentiality and potentially non-zero adaptivity gain. 
In~\cite{Hayashi09}, the error exponent corresponding to the class of NN and NA policies 
were fully characterized for the problem of active binary hypothesis testing with fixed sample size.  In the Bayesian context, the result 
of \cite{Hayashi09} regarding the error exponent of the class of NN policies coincides with our Corollary~\ref{MaxE},
while the full characterization of the error exponent corresponding to the
class of NA policies in \cite{Hayashi09}, strengthens Corollary~\ref{MaxE-NA} 
in the binary case. In particular, it is shown that in the binary hypothesis testing setup
$E_{NN}=E_{NA}$, hence, establishing zero adaptivity gain among non-sequential policies.
For the special case of  channel coding with feedback\footnote{The problem of channel coding with feedback can be interpreted as a special case of active hypothesis testing (See~\cite{Asilomar10} for more details).} with two messages, 
the above result, i.e., the zero adaptivity gain among non-sequential policies,  
was established in~\cite{Berlekamp64, NakibogluPhD}.

\subsection{Numerical Example}
\label{sec:exmp}

Consider the active binary hypothesis testing problem with additive Gaussian noisy observations under two actions $a$ and $b$ shown in Fig.~\ref{fig:binary}. In this example, the observation noise associated with 
actions $a$ and $b$ are such that they add unequal noise to the hypotheses. 
In the remainder of this subsection, we compare the performance of all considered policies for this example.

\begin{figure}[htp]
\centering
\psfrag{x}{\scriptsize{$\theta$}}
\includegraphics[width=0.4\textwidth]{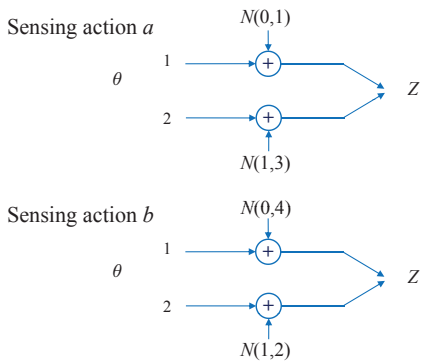}
\caption{Active binary hypothesis testing problem with additive Gaussian noisy observations.}
\label{fig:binary}
\end{figure}

Table~\ref{binary} compares the performance bounds of the considered policies
for the example of Fig.\ref{fig:binary}.

{\scriptsize{
\begin{table}[htp]
\center
\caption{Comparison of performance bounds for the example of Fig.\ref{fig:binary}.}
\label{binary}
\begin{tabular}{ccc}
  \toprule
  & Sequential & Non-sequential \\
  \midrule
  \vspace{0.05 in} 
  Adaptive &  $\log L / 2.98$ & $\lesssim \log L / 1.89$  \\
  \midrule
  \vspace{0.05 in} 
  Non-adaptive & $\log L / 2.27$  & $2 \log L / 1.78 $  \\
  \bottomrule
\end{tabular}
\end{table}
}}

\section{Discussion and Future Work}
\label{Discussion}

In this paper, we considered the problem of active hypothesis testing and we analyzed 
the gain of sequential and adaptive selection of actions.

Our analysis assumes two technical conditions. However, it seems to us that Assumption~\ref{Jump} is for ease of our proofs. As part 
of our future work, we believe that 
standard techniques as in~\cite{Burnashev80, Bentkus08} can be applied to generalize the bounds when 
Assumption~\ref{Jump} does not hold. 
We also note the results obtained in~\cite{Lorden77} and~\cite{Chernoff59, ISIT11} have been shown in~\cite{Dragalin99} and \cite{Sundaresan12}, respectively, to extend to higher moment 
characterization of the optimal (sequential) sample size.  Similar extension in the 
context of sequential and non-adaptive policies seem to follow naturally and
is important area of future investigation.

In our analysis in this paper, we only investigated asymptotic performance in $L$ and
the complementary role of asymptotic analysis in $M$ was neglected. In particular, we have only 
identified the zero-rate characterization of error exponent; while for a full characterization in which 
error exponent is traded off with information acquisition rate, we would need an asymptotic characterization of the problem both in $L$ and $M$.  
Although we have partially addressed this problem in~\cite{HypJournal} for the class of sequential policies, 
the full characterization of the performance bounds in $L$ and $M$ for 
all types of policies defined in this paper remains an important area of future work.

\appendix



\subsection{Theorem~\ref{VNN}, non-sequential non-adaptive policy} \label{appendix:NN}

In this subsection, we show that
\begin{align}
V_{NN} (\boldsymbol{\rho}) &\le 
\frac{\log L - \min \limits_{i,j \in \Omega} \log\frac{\rho_i}{\rho_j}}
{\hat{D}} + o(\log L),\\
\label{VNN01}
V_{NN} (\boldsymbol{\rho}) &\ge 
\frac{\log L - \max \limits_{i,j \in \Omega} \log\frac{\rho_i}{\rho_j}}
{\hat{D}} - o(\log L),
\end{align} 
where
\begin{align}
\label{DenDef}
\hat{D} = \max \limits_{\boldsymbol{\lambda} \in \Lambda(\mathcal{A})} \min \limits_{i \in \Omega}\min \limits_{j \neq i} 
\max \limits_{\alpha \in [0,1]} \sum_{a \in \mathcal{A}} \lambda_a (1-\alpha) D_{\alpha} (q_i^a ||q_j^a).
\end{align}
%
%


Suppose $\hat{\boldsymbol{\lambda}} \in \Lambda(\mathcal{A})$ achieves the maximum in (\ref{DenDef}).
Let $\pi_{NN}$ be a non-sequential non-adaptive policy that collects $\hat{n}$
observation samples and selects sensing actions according to the randomized rule $\hat{\boldsymbol{\lambda}}$.
The expected total cost under this policy is $\hat{n}+ L \Pe$.
Next we find an upper bound for $\Pe$. 
Let $\mathcal{Z}_i(n)=\left\{Z^{n}:\rho_i(n)\ge\rho_j(n) \text{ for all } j \in \Omega \right\}$
and $e_{ij}(n)=P(\{Z^{n}:\rho_i(n)<\rho_j(n)\}|\theta=i)$.
\begin{align}
\label{VNN04}
  \nonumber
	\Pe	&= \sum_{i=1}^M \rho_i P(\cup_{j\neq i}\{Z^{{\hat{n}}}:\rho_i(\hat{n})<\rho_j(\hat{n})\}|\theta=i) \\ \nonumber
	&\le \sum_{i=1}^M \rho_i \sum_{j \neq i} e_{ij}({\hat{n}})\\
	&\le (M-1) \max_{i,j \in \Omega} e_{ij}({\hat{n}}).
\end{align}

From~\eqref{VNN04} and Lemma~\ref{Blahut} in Appendix~\ref{app:facts}, we obtain 
\begin{align*}
\lefteqn{\Pe \le (M-1) \times}\\
&\exp\bigg(-\hat{n} (1-\alpha) \sum \limits_{a \in \mathcal{A}} \hat{\lambda}_a D_{\alpha}(q_i^a||q_j^a) - \min_{i,j\in\Omega} \log\frac{\rho_i}{\rho_j} + o(\hat{n})\bigg).
\end{align*}
We can select $\hat{n}$ as
\begin{align}
\label{VNN03}
\hat{n}=\left(\log L + \log (M-1) - \min \limits_{i,j \in \Omega} \log\frac{\rho_i}{\rho_j} + o(\log L) \right)/\hat{D}
\end{align}
such that $\Pe = O(\frac{1}{L})$, 
and hence,
\begin{align*}
V_{NN} \le {\hat{n}} + L \Pe \le {\hat{n}} + 1 = \frac{\log L - \min \limits_{i,j \in \Omega} \log\frac{\rho_i}{\rho_j}}{\hat{D}} + o(\log L).
\end{align*}
This completes the proof of upper bound.
Next the proof of lower bound is given. 

%
Consider a policy $\pi_{NN}$ that collects $n$
observation samples according to $\boldsymbol{\lambda} \in \Lambda(\mathcal{A})$.
%
We have
\begin{align}
\label{VNNU01}
\nonumber
	\Pe	&= \sum_{i=1}^M \rho_i P(\cup_{j\neq i}\{Z^{n}:\rho_i(n)<\rho_j(n)\}|\theta=i) \\
	&\ge \rho_i e_{ij} + \rho_j e_{ji} \ \ \text{ for any } i,j \in \Omega.
\end{align}
From \eqref{VNNU01} and Lemma~\ref{Blahut} in Appendix~\ref{app:facts}, 
a lower bound is obtained for the expected total cost under policy $\pi_{NN}$.
The lower bound for $V_{NN}$ is obtained by minimizing over the choices of $n$ and~$\boldsymbol{\lambda}$.

\subsection{Theorem~\ref{VSN}, sequential non-adaptive policy} 
\label{appendix:SN}

In this subsection, we show that
\begin{align}
\label{VSN01}
V_{SN} (\boldsymbol{\rho}) &\le \min \limits_{\boldsymbol{\lambda} \in \Lambda(\mathcal{A})}
\sum_{i=1}^{M} \rho_i \frac{\log L - \min \limits_{k \neq i} \log\frac{\rho_i}{\rho_k}}{R(i,\boldsymbol{\lambda})} + o(\log L),\\
\label{VSN02}
V_{SN} (\boldsymbol{\rho}) &\ge \min \limits_{\boldsymbol{\lambda} \in \Lambda(\mathcal{A})}
\sum_{i=1}^{M} \rho_i \frac{\log L - \max \limits_{k \neq i} \log\frac{\rho_i}{\rho_k}}{R(i,\boldsymbol{\lambda})} - o(\log L).
\end{align}


In contrast to the passive case, the observations in the active case (either adaptive or non-adaptive)
are not necessarily identical over time.
Therefore the analysis of~\cite{Lorden77} for sequential passive hypothesis testing (which is based on the law of large number and results for random walks) is not applicable to the problem of sequential non-adaptive hypothesis testing.

Suppose $\hat{\boldsymbol{\lambda}} \in \Lambda(\mathcal{A})$ achieves the minimum in (\ref{VSN01}).
The upper bound \eqref{VSN01} is achieved by a policy that selects sensing actions according to $\hat{\boldsymbol{\lambda}}$
and stops sampling at $$\tau:=\min\{n:\max \limits_{i \in \Omega} \rho_i(n) \ge 1-L^{-1}\}.$$
Let $\tau_i$, $i\in\Omega$, be Markov stopping times defined as follows:
\begin{align}
\label{Deftaui}
\tau_i &:=\min \left\{ n: \min_{j \neq i} \frac{\rho_i(n)}{\rho_j(n)} \ge \frac{1-L^{-1}}{L^{-1}/(M-1)} \right\}.
\end{align}
Note that by definition
\begin{align*}
(M-1) \rho_i(\tau_i) &\ge \sum_{j\neq i} \rho_j(\tau_i) \frac{1-L^{-1}}{L^{-1}/(M-1)}\\
&= (M-1) (1-\rho_i(\tau_i)) \frac{1-L^{-1}}{L^{-1}}.
\end{align*}
This implies that $\rho_i(\tau_i) \ge 1-L^{-1}$ and hence, $\tau \le \tau_i$ for all $i\in\Omega$.
From~\eqref{Obj2}, total cost under the above policy can be written as 
\begin{align}
\nonumber
\label{UBVpi}
V(\boldsymbol{\rho}) &= \mathbb{E}[\tau] + L [1-\max_{j\in\Omega} \rho_j(\tau)]\\
&\le \mathbb{E}[\tau] + 1 \nonumber \\
&= \sum_{i=1}^{M} \rho_i \mathbb{E}[\tau | \theta=i] + 1 \nonumber \\
& \le \sum_{i=1}^{M} \rho_i \mathbb{E}[\tau_i | \theta=i] + 1,
\end{align}
where $\boldsymbol{\rho}=[\rho_1,\rho_2, \ldots, \rho_M]=[\rho_1(0),\rho_2(0), \ldots \rho_M(0)]$ 
and the last inequality follows from the fact that $\tau \leq \tau_i$, $\forall i\in\Omega$. 

Next we find an upper bound for $\mathbb{E}[\tau_i | \theta=i]$, $i\in\Omega$.
Before we proceed, we introduce the following notation to facilitate the proof:
\begin{align*}
T_i &:= \log\frac{1-{L^{-1}}}{L^{-1}/(M-1)} - \min \limits_{k \neq i} \log\frac{\rho_i}{\rho_k}.
\end{align*}
%
Let $\iota := (\log L)^{-\frac{1}{4}}$. We have
\begin{align}
\label{UBtaui}
\nonumber
\mathbb{E} [ \tau_i | \theta=i ] 
&= \sum_{n=0}^{\infty} P(\{\tau_i > n\} | \theta=i)\\
\nonumber
&\le 1 + \frac{T_i}{R(i,\hat{\boldsymbol{\lambda}})}(1+\iota) 
+ \hspace*{-0.2in} \sum_{n: n> \frac{T_i}{R(i,\hat{\boldsymbol{\lambda}})}(1+\iota)} \hspace*{-0.29in} P(\{\tau_i > n\} | \theta=i)\\
\nonumber
&\stackrel{(a)}{\le} \frac{T_i}{R(i,\hat{\boldsymbol{\lambda}})} + o(\log L)\\ 
&\le \frac{\log L - \min \limits_{k \neq i} \log\frac{\rho_i}{\rho_k}}{R(i,\hat{\boldsymbol{\lambda}})} + o(\log L),
\end{align}where inequality $(a)$ follows from the fact that $\iota =(\log L)^{-\frac{1}{4}}$ and by Lemma~\ref{expdecay} in Appendix~\ref{app:facts}. 
Now from~\eqref{UBVpi} and~\eqref{UBtaui}, we have the assertion of the theorem.

Next we provide the proof of lower bound~\eqref{VSN02} which follows closely the proof of Theorem~2 in~\cite{Chernoff59}.

From upper bound~\eqref{VSN01} we know that the total cost under the optimal policy is $O(\log L)$.
This implies that the $\Pe$ of the optimal policy is $O(\frac{\log L}{L})$.
Hence, without loss of generality in our computation of the lower bound, 
we can restrict the set of policies to those whose average probability of making an error is $O(\frac{\log L}{L})$.   

Let $\pi_{SN}$ denote a sequential policy that
selects sensing actions according to $\boldsymbol{\lambda} \in \Lambda(\mathcal{A})$
and stops sampling whenever $\Pe \le \epsilon$.
For all $i \in \Omega$, let 
\begin{align}
\label{VSNL00}
T_i := (1-\delta) \frac{\log\frac{1}{\epsilon} - \max \limits_{k \neq i} \log\frac{\rho_i}{\rho_k}}{R(i,\boldsymbol{\lambda})+\delta}.
\end{align}

Under policy $\pi_{SN}$,
\begin{align}
\label{VSNL02}
\nonumber
\lefteqn{P(\left\{\tau < T_i\right\} | \theta=i)}\\
\nonumber
&= P\left(\left\{\tau < T_i\right\} \cap \bigcap_{j\neq i}\left\{\frac{\rho_i(\tau)}{\rho_j(\tau)} \ge (\frac{1}{\epsilon})^{1-\delta}\right\} | \theta=i\right) \\
\nonumber
&\hspace{.15in} + P\left(\left\{\tau < T_i\right\} \cap \bigcup_{j\neq i} \left\{\frac{\rho_i(\tau)}{\rho_j(\tau)} < (\frac{1}{\epsilon})^{1-\delta}\right\} | \theta=i\right) \\ 
\nonumber
&\stackrel{(a)}{\le} \frac{(\log\xi)^2}{T_i \delta^2} + \sum_{j\neq i} 
P\left( \left\{\frac{\rho_i(\tau)}{\rho_j(\tau)} < (\frac{1}{\epsilon})^{1-\delta}\right\} | \theta=i\right) \\
&\stackrel{(b)}{\le} \frac{(\log\xi)^2}{T_i \delta^2} + (M-1) \epsilon^{\delta} \left(\frac{1}{\rho_i} + \frac{1}{\min_{j\neq i} \rho_j}\right),
\end{align}
where $(a)$ follows from Lemma~\ref{B1event} in Appendix~\ref{app:facts} and the union bound;
and $(b)$ follows from Lemma~\ref{B2event} in Appendix~\ref{app:facts}.

The expected total cost under policy $\pi_{SN}$ is lower bounded as
\begin{align*}
 \Expt[\tau] + L \Pe &\ge \sum_{i=1}^M \rho_i \Expt[\tau | \theta=i] \\
 &= \sum_{i=1}^M \rho_i \Expt[\tau \indc_{\{\tau\ge T_i\}} + \tau \indc_{\{\tau< T_i\}} | \theta=i] \\
 &\ge \sum_{i=1}^M \rho_i T_i P(\tau\ge T_i | \theta=i) \\
 &\ge \sum_{i=1}^M \rho_i T_i \Big(1 - \frac{(\log\xi)^2}{T_i \delta^2} - \frac{2 \epsilon^{\delta} M}{\min_{j\in\Omega} \rho_j}\Big).
\end{align*}
For $\delta=(\log\frac{1}{\epsilon})^{-\frac{1}{4}}$, the lower bound simplifies to 
\begin{align*}
 \Expt[\tau] + L \Pe 
 &\ge \sum_{i=1}^M \rho_i \frac{\log\frac{1}{\epsilon} - \max \limits_{k \neq i} \log\frac{\rho_i}{\rho_k}}{R(i,\boldsymbol{\lambda})} - o(\log\frac{1}{\epsilon})\\
 &\ge \sum_{i=1}^M \rho_i \frac{\log L - \max \limits_{k \neq i} \log\frac{\rho_i}{\rho_k}}{R(i,\boldsymbol{\lambda})} - o(\log L),
\end{align*}
where the last inequality follows from the fact that for an optimal policy, $\epsilon=O(\frac{\log L}{L})$.
The lower bound for $V_{SN}$ is obtained
by minimizing over the choice of $\boldsymbol{\lambda}$.


\subsection{Theorem~\ref{VSA}, sequential adaptive policy}
\label{appendix:SA}

We have
\begin{align}
\label{VSAub}
V_{SA} (\boldsymbol{\rho}) &\le 
\sum_{i=1}^{M} \rho_i \frac{\log L - \min \limits_{k \neq i} \log\frac{\rho_i}{\rho_k}}{R(i,\boldsymbol{\lambda}^*_i)} + o(\log L),\\
\label{VSAlb}
V_{SA} (\boldsymbol{\rho}) &\ge
\sum_{i=1}^{M} \rho_i \frac{\log L - \max \limits_{k \neq i} \log\frac{\rho_i}{\rho_k}}{R(i,\boldsymbol{\lambda}^*_i)} - o(\log L).
\end{align} 
%
%
The upper bound was proved in~\cite[Prop.~3]{HypJournal}. 
The proof of the lower bound relies on a generalization of Theorem~2 in~\cite{Chernoff59} and is provided next. 

From upper bound~\eqref{VSAub} we know that the total cost under the optimal policy is $O(\log L)$.
This implies that the error probability $\Pe$ of the optimal policy is $O(\frac{\log L}{L})$.

Let $\pi_{SA}$ denote a sequential policy that
stops sampling whenever $\Pe \le \epsilon$.
For all $i \in \Omega$, let 
\begin{align}
\label{VSAL00}
T^*_i := (1-\delta) \frac{\log\frac{1}{\epsilon} - \max \limits_{k \neq i} \log\frac{\rho_i}{\rho_k}}{R(i,\boldsymbol{\lambda}^*_i)+\delta}.
\end{align}

Under policy $\pi_{SA}$,
\begin{align}
\label{VSAL02}
\nonumber
\lefteqn{P(\left\{\tau < T^*_i\right\} | \theta=i)}\\
\nonumber
&= P\left(\left\{\tau < T^*_i\right\} \cap \bigcap_{j\neq i}\left\{\frac{\rho_i(\tau)}{\rho_j(\tau)} \ge (\frac{1}{\epsilon})^{1-\delta}\right\} | \theta=i\right) \\
\nonumber
&\hspace{.15in} + P\left(\left\{\tau < T^*_i\right\} \cap \bigcup_{j\neq i} \left\{\frac{\rho_i(\tau)}{\rho_j(\tau)} < (\frac{1}{\epsilon})^{1-\delta}\right\} | \theta=i\right) \\ 
\nonumber
&\stackrel{(a)}{\le} \frac{(\log\xi)^2}{T^*_i \delta^2} + \sum_{j\neq i} 
P\left( \left\{\frac{\rho_i(\tau)}{\rho_j(\tau)} < (\frac{1}{\epsilon})^{1-\delta}\right\} | \theta=i\right) \\
&\stackrel{(b)}{\le} \frac{(\log\xi)^2}{T^*_i \delta^2} + (M-1) \epsilon^{\delta} \left(\frac{1}{\rho_i} + \frac{1}{\min_{j\neq i} \rho_j}\right),
\end{align}
where $(a)$ follows from Lemma~\ref{B3event} in Appendix~\ref{app:facts} and the union bound;
and $(b)$ follows from Lemma~\ref{B2event} in Appendix~\ref{app:facts}.

The expected total cost under policy $\pi_{SA}$ is lower bounded as
\begin{align*}
 \Expt[\tau] + L \Pe &\ge \sum_{i=1}^M \rho_i \Expt[\tau | \theta=i] \\
 &= \sum_{i=1}^M \rho_i \Expt[\tau \indc_{\{\tau\ge T^*_i\}} + \tau \indc_{\{\tau< T^*_i\}} | \theta=i] \\
 &\ge \sum_{i=1}^M \rho_i T^*_i P(\tau\ge T^*_i | \theta=i) \\
 &\ge \sum_{i=1}^M \rho_i T^*_i \Big(1 - \frac{(\log\xi)^2}{T^*_i \delta^2} - \frac{2 \epsilon^{\delta} M}{\min_{j\in\Omega} \rho_j}\Big).
\end{align*}
For $\delta=(\log\frac{1}{\epsilon})^{-\frac{1}{4}}$, the lower bound simplifies to 
\begin{align*}
 \Expt[\tau] + L \Pe 
 &\ge \sum_{i=1}^M \rho_i \frac{\log\frac{1}{\epsilon} - \max \limits_{k \neq i} \log\frac{\rho_i}{\rho_k}}{R(i,\boldsymbol{\lambda}^*_i)} - o(\log\frac{1}{\epsilon})\\
 &\ge \sum_{i=1}^M \rho_i \frac{\log L - \max \limits_{k \neq i} \log\frac{\rho_i}{\rho_k}}{R(i,\boldsymbol{\lambda}^*_i)} - o(\log L),
\end{align*}
where the last inequality follows from the fact that for an optimal policy, $\epsilon=O(\frac{\log L}{L})$.


\begin{remarks}
The result above is in agreement with Theorem~2 in~\cite{Chernoff59}
and Theorem~4 in~\cite{NitinawaratArxiv}.
\end{remarks}

\subsection{Proposition~\ref{VNA}, non-sequential adaptive policy}

In this subsection, we show that
\begin{align*}
V_{NA} (\boldsymbol{\rho}) &\ge
\frac{\log L - \max \limits_{k \neq i} \log\frac{\rho_i}{\rho_k}}
{\min \limits_{i \in \Omega} \max \limits_{\boldsymbol{\lambda} \in \Lambda(\mathcal{A})} R(i,\boldsymbol{\lambda})} - o(\log L).
\end{align*} 

\begin{IEEEproof}

Let $\pi_{NA}$ be a non-sequential adaptive policy that collects $n$
observation samples.
%
Consider an arbitrary $\delta>0$ and let 
$$\epsilon_i=\frac{1}{\exp\left(n(R(i,\boldsymbol{\lambda}^*_i) + \delta) 
+ \max \limits_{k \neq i} \log\frac{\rho_i}{\rho_k}\right)+1}.$$

We have
\begin{align}
\label{VNNL01}
	\Pe 
		 &\ge \sum_{i=1}^M \rho_i \Expt[1-\rho_i(n)| \theta=i] P(\mathcal{Z}_i|\theta=i),
\end{align}
where
\begin{align}		 
\label{VNNL02}
		 \Expt[1-\rho_i(n)| \theta=i] &\ge \epsilon_i P(1-\rho_i(n) \ge \epsilon_i | \theta=i).
\end{align}
%
Let $\hat{j} = \argmin \limits_{j \neq i} \sum_{t=0}^{n-1} \Expt[\log\frac{q_i^{A(t)}(Z)}{q_j^{A(t)}(Z)}|\theta=i]$
where actions $\{A(t)\}_{t=0}^{n-1}$ are selected according to $\pi_{NA}$.
%
\begin{align}
\label{VNNL03}
\nonumber
		 \lefteqn{P(1-\rho_i(n) < \epsilon_i | \theta=i)}\\
\nonumber
		 &= P\left(\log\frac{\rho_i(n)}{1-\rho_i(n)} > \log\frac{1-\epsilon_i}{\epsilon_i} | \theta=i\right) \\
\nonumber
		 &\le P\left(\cap_{j\neq i} \left\{\log\frac{\rho_i(n)}{\rho_j(n)} > \log\frac{1-\epsilon_i}{\epsilon_i}\right\} | \theta=i\right)\\
\nonumber
     &\stackrel{(a)}{\le} P\left(\left\{\log\frac{\rho_i(n)}{\rho_{\hat{j}}(n)} - \Expt[\log\frac{\rho_i(n)}{\rho_{\hat{j}}(n)}] \right. \right.\\
\nonumber
     &\hspace*{0.55in} \left. \left. > \log\frac{1-\epsilon_i}{\epsilon_i} - \max \limits_{k \neq i} \log\frac{\rho_i}{\rho_k} - n R(i,\boldsymbol{\lambda}_i^*) \right\} | \theta=i \right) \\
\nonumber
		 &\le P\left(\log\frac{\rho_i(n)}{\rho_{\hat{j}}(n)} - \Expt[\log\frac{\rho_i(n)}{\rho_{\hat{j}}(n)}] > n \delta | \theta=i\right) \\
		 &\stackrel{(b)}{\le} \exp(-n\delta^2/(\log\xi)^2),
\end{align} 
where $(a)$ follows from the fact that given $\{\theta=i\}$,
\begin{align}
\label{Edrift}
\nonumber
\Expt[\log\frac{\rho_i(n)}{\rho_{\hat{j}}(n)}]
&= \log\frac{\rho_i}{\rho_{\hat{j}}} + \sum_{t=0}^{n-1} \Expt[\log\frac{\rho_i(t+1)}{\rho_{\hat{j}}(t+1)} - \log\frac{\rho_i(t)}{\rho_{\hat{j}}(t)}]\\
\nonumber
&= \log\frac{\rho_i}{\rho_{\hat{j}}} + \sum_{t=0}^{n-1} \Expt[\log\frac{q_i^{A(t)}(Z)}{q_{\hat{j}}^{A(t)}(Z)}]\\
&\le \max \limits_{k \neq i} \log\frac{\rho_i}{\rho_k} + n \min \limits_{j \neq i} \sum_{a \in \mathcal{A}} \lambda^*_{ia} D(q_i^{a}||q_j^{a}),  
\end{align} 
and $(b)$ follows from Fact~\ref{McDiarmid}.

Similarly, it can be shown that
\begin{align}
	\label{VNNL04}
  P(\mathcal{Z}_i^c|\theta=i) \le \exp(-n(R(i,\boldsymbol{\lambda}^*_i))^2/(\log\xi)^2).
\end{align}

Combining (\ref{VNNL01})--(\ref{VNNL04}) and minimizing the bound over $n$,
we have the assertion of the proposition.
%
%
\end{IEEEproof}

\subsection{Technical Background}
\label{app:facts}

In this appendix, we provide some preliminary facts and lemmas
which are technical and only helpful in proving the main results of the paper.

\begin{fact}[Kolmogorov's Maximal Inequality~\cite{Billingsley}]
\label{Kolmogorov}
Suppose $X_t$ for $t=1,2,\ldots$, be independent random variables with $\mathbb{E}[X_t]=0$ and $Var(X_t)<\infty$.
Let $S_n=\sum_{t=1}^{n} X_t$. Then
$$P\left(\max_{0\le n \le N} |S_n| > x \right) \le \frac{Var(S_N)}{x^2} = \frac{\sum_{t=1}^N Var(X_t)}{x^2}.$$
\end{fact}

\begin{fact}[McDiarmid's Inequality~\cite{McDiarmid89}]
\label{McDiarmid}
Let $\mathbf{X} = (X_1, \ldots, X_n)$ be a family of independent random variables with $X_k$ taking values in a set $\mathcal{X}_k$ for each $k$. Suppose a real-valued function $f$ defined on $\Pi_{k=1}^n \mathcal{X}_k$ satisfies 
$|f(\mathbf{x}) - f(\mathbf{x}') | \le c_k$,
whenever the vectors $\mathbf{x}$ and $\mathbf{x}'$ only differ in the $k$-th coordinate.
Then for any $\nu > 0$,
\begin{align*}
	P(f(\mathbf{X}) - \mathbb{E}[f(\mathbf{X})] \ge  \nu) &\le e^{-2 \nu^2 / \sum_{k=1}^n c_k^2},\\
	P(f(\mathbf{X}) - \mathbb{E}[f(\mathbf{X})] \le -\nu) &\le e^{-2 \nu^2 / \sum_{k=1}^n c_k^2}.
\end{align*}
\end{fact}

\begin{lemma}
\label{Blahut}
Consider a policy that collects observation samples according to a randomized rule $\boldsymbol{\lambda}$. Under this policy and for all $i,j \in \Omega$, and $\alpha \in [0,1]$,  
\begin{align*}
\max\left\{e_{ij}(n),e_{ji}(n)\right\} &\le \exp\bigg(-n (1-\alpha) \sum \limits_{a \in \mathcal{A}} \lambda_a D_{\alpha}(q_i^a||q_j^a)\\ &\hspace*{.37in} - \min\big\{\log\frac{\rho_i}{\rho_j}, \log\frac{\rho_j}{\rho_i}\big\} + o(n)\bigg),\\
\max\left\{e_{ij}(n),e_{ji}(n)\right\} &\ge \exp\bigg(-n (1-\alpha) \sum \limits_{a \in \mathcal{A}} \lambda_a D_{\alpha}(q_i^a||q_j^a)\\ &\hspace*{.34in} - \max\big\{ \log\frac{\rho_i}{\rho_j}, \log\frac{\rho_j}{\rho_i}\big\} - o(n) \bigg).
\end{align*}
\end{lemma}

The proof of Lemma~\ref{Blahut} follows closely the proof of Theorem~9 in~\cite{Blahut74}. 


\begin{lemma}
\label{expdecay} 
Given any $\iota > 0$ and for $n > \frac{T_i}{R(i,\hat{\boldsymbol{\lambda}})}(1+\iota)$, 
we have $P(\{\tau_i > n\} | \theta=i) \le (M-1) e^{- b(\iota) n}$ where 
$$b(\iota)=\frac{2 \iota^2}{(1+\iota)^2} \bigg(\frac{R(i,\hat{\boldsymbol{\lambda}})}{2 \log\xi}\bigg)^2.$$
\end{lemma}

\begin{IEEEproof}[Proof of Lemma~\ref{expdecay}]

Let $B_{ij}(n)$ be an event in the probability space defined as follows:
\begin{align*}
B_{ij}(n) &:= \left\{ \log\frac{\rho_i(n)}{\rho_j(n)} < \log\frac{1-L^{-1}}{L^{-1}/(M-1)} \right\}.
\end{align*}

By construction~\eqref{Deftaui},
\begin{align}
\label{Mc3}
\nonumber
P(\{\tau_i > n\} | \theta=i) &\le P(\cup_{j \neq i} B_{ij}(n) | \theta=i)\\
&\le \sum_{j \neq i} P( B_{ij}(n) | \theta=i).
\end{align}

Furthermore, we have
{\allowdisplaybreaks{
\begin{align}
\nonumber
\lefteqn{P(B_{ij}(n) | \theta=i)}\\ 
\nonumber
&= P\Big ( \Big \{ \log\frac{\rho_i(n)}{\rho_j(n)} - \mathbb{E}[\log\frac{\rho_i(n)}{\rho_j(n)}] <\\
\nonumber
&\hspace*{.9in} \log\frac{1-L^{-1}}{L^{-1}/(M-1)} - \mathbb{E}\big[\log\frac{\rho_i(n)}{\rho_j(n)}\big] \Big \} \big | \theta=i \Big ) \\
\nonumber
&= P\Big ( \Big \{ \log\frac{\rho_i(n)}{\rho_j(n)} - \mathbb{E}[\log\frac{\rho_i(n)}{\rho_j(n)}] <\\
\nonumber
&\hspace*{.17in} \log\frac{1-L^{-1}}{L^{-1}/(M-1)} - \mathbb{E}\big[\log\frac{\rho_i}{\rho_j} + \sum_{t=0}^{n-1}\log\frac{q_i^{A(t)}}{q_j^{A(t)}} \big] \Big \} \big | \theta=i \Big ) \\
\nonumber
&\le P\Big ( \Big \{ \log\frac{\rho_i(n)}{\rho_j(n)} - \mathbb{E}[\log\frac{\rho_i(n)}{\rho_j(n)}] <\\
\nonumber 
&\hspace*{.38in} \log\frac{1-L^{-1}}{L^{-1}/(M-1)} -\min_{k \neq i} \log\frac{\rho_i}{\rho_k} - n R(i,\hat{\boldsymbol{\lambda}}) \Big \} \big | \theta=i \Big )\\
\label{MC3a}
&= P\Big ( \Big \{ \log\frac{\rho_i(n)}{\rho_j(n)} - \mathbb{E}[\log\frac{\rho_i(n)}{\rho_j(n)}] <
T_i - n R(i,\hat{\boldsymbol{\lambda}}) \Big \} \big | \theta=i \Big ).
\end{align}}}

For any $a,\hat{a}\in\mathcal{A}$ and $i,j\in\Omega$, we have $\left|\log\frac{q_i^a}{q_j^a}-\log\frac{q_i^{\hat{a}}}{q_j^{\hat{a}}}\right| \le 2 \log\xi$.
For $k=1,2,\ldots,n$, let $X_k=\log\frac{q_i^{A(k-1)}}{q_j^{A(k-1)}}$ and $\boldsymbol{X}=[X_1,X_2,\ldots,X_n]$.
Define function $f(\boldsymbol{X})=\log\frac{\rho_i}{\rho_j}+\sum_{k=1}^n X_k=\log\frac{\rho_i(n)}{\rho_j(n)}$.
From~\eqref{Mc3}, \eqref{MC3a}, and Fact~\ref{McDiarmid}, and for 
$n > \frac{T_i}{R(i,\hat{\boldsymbol{\lambda}})}(1+\iota)$, we have

\begin{align}
\nonumber
&P(\{\tau_i > n\} | \theta=i) \\
\nonumber
&\le (M-1) \exp \left(-2 n \bigg(\frac{R(i,\hat{\boldsymbol{\lambda}})}{2\log\xi}\bigg)^2 \bigg(1-\frac{1}{n}\frac{T_i}{R(i,\hat{\boldsymbol{\lambda}})}\bigg)^2 \right) \\
\nonumber
&\le (M-1) \exp \left(- n \frac{2 \iota^2}{(1+\iota)^2} \bigg(\frac{R(i,\hat{\boldsymbol{\lambda}})}{2\log\xi}\bigg)^2 \right).
\end{align}
\end{IEEEproof}

\begin{lemma}
\label{B2event}
Consider a sequential policy $\pi$ that selects the stopping time $\tau$ such that $\Pe \le \epsilon$.
For any $i,j\in\Omega$, we have
\begin{align*}
P\left(\left\{\frac{\rho_i(\tau)}{\rho_j(\tau)} < (\frac{1}{\epsilon})^{1-\delta}\right\} | \theta=i\right) \le \epsilon^{\delta} \left(\frac{1}{\rho_i} + \frac{1}{\rho_j}\right).
\end{align*}
\end{lemma}

\begin{IEEEproof}
The proof follows closely the proof of Lemma~4 in~\cite{Chernoff59}. 
Let $\hat{\theta}=d(A^{\tau},Z^{\tau})$ denote the final declaration under policy $\pi$. We have
\begin{align*}
\lefteqn{P\left(\left\{\frac{\rho_i(\tau)}{\rho_j(\tau)} < (\frac{1}{\epsilon})^{1-\delta}\right\} | \theta=i\right)}\\
&=P\left(\left\{\frac{\rho_i(\tau)}{\rho_j(\tau)} < (\frac{1}{\epsilon})^{1-\delta}\right\} \cap \left\{\hat{\theta}=i \right\} | \theta=i\right)\\
&\hspace*{.1in} + P\left(\left\{\frac{\rho_i(\tau)}{\rho_j(\tau)} < (\frac{1}{\epsilon})^{1-\delta}\right\} \cap \left\{\hat{\theta} \neq i \right\} | \theta=i\right)\\
&\le (\frac{1}{\epsilon})^{1-\delta} P\left(\left\{\hat{\theta}=i \right\} | \theta=j\right)
+ P\left(\left\{\hat{\theta}\neq i \right\} | \theta=i\right)\\
&\stackrel{(a)}{\le} (\frac{1}{\epsilon})^{1-\delta} P\left(\left\{\hat{\theta}\neq j \right\} | \theta=j\right)
+ P\left(\left\{\hat{\theta}\neq i \right\} | \theta=i\right)\\
&\le (\frac{1}{\epsilon})^{1-\delta} \frac{\epsilon}{\rho_j} + \frac{\epsilon}{\rho_i} \\
&=  \frac{\epsilon^{\delta}}{\rho_j} + \frac{\epsilon}{\rho_i}\\
&\le \epsilon^{\delta} \left(\frac{1}{\rho_i} + \frac{1}{\rho_j}\right),
\end{align*}
where $(a)$ follows from the fact that under policy $\pi$ and for all $i\in\Omega$,
\begin{align*}
P\left(\left\{\hat{\theta}\neq i \right\} | \theta=i\right) &\le \frac{1}{\rho_i} \sum_{k=1}^M \rho_k P\left(\left\{\hat{\theta}\neq k \right\} | \theta=k\right)\\
&= \frac{1}{\rho_i} \Pe\\
&\le \frac{\epsilon}{\rho_i}.
\end{align*}
\end{IEEEproof}

\begin{lemma}
\label{B1event}
Consider a sequential policy $\pi$ that selects sensing actions 
according to $\boldsymbol{\lambda} \in \Lambda(\mathcal{A})$
and selects the stopping time $\tau$ such that $\Pe \le \epsilon$.
We have
\begin{align*}
P\left(\left\{\tau < T_i\right\} \cap \bigcap_{j\neq i}\left\{\frac{\rho_i(\tau)}{\rho_j(\tau)} \ge (\frac{1}{\epsilon})^{1-\delta}\right\} | \theta=i\right) \le \frac{(\log\xi)^2}{T_i \delta^2},
\end{align*}
where $T_i$ is as defined in \eqref{VSNL00}.
\end{lemma}

The proof of Lemma~\ref{B1event} follows closely the proof of Lemma~5 in~\cite{Chernoff59}.

\begin{lemma}
\label{B3event}
Consider a sequential policy $\pi$ that selects the stopping time $\tau$ such that $\Pe \le \epsilon$.
We have
\begin{align*}
P\left(\left\{\tau < T^*_i\right\} \cap \bigcap_{j\neq i}\left\{\frac{\rho_i(\tau)}{\rho_j(\tau)} \ge (\frac{1}{\epsilon})^{1-\delta}\right\} | \theta=i\right) \le \frac{(\log\xi)^2}{T^*_i \delta^2},
\end{align*}
where $T^*_i$ is as defined in \eqref{VSAL00}.
\end{lemma}

\begin{IEEEproof}
The proof follows closely the proof of Lemma~5 in~\cite{Chernoff59}.
We have
\begin{align*}
\lefteqn{P\bigg(\Big\{\tau < T^*_i\Big\} \cap \bigcap_{j\neq i}\Big\{\frac{\rho_i(\tau)}{\rho_j(\tau)} \ge (\frac{1}{\epsilon})^{1-\delta}\Big\} | \theta=i\bigg)} \\
&\le P\Big(\Big\{\min_{n}: \log\frac{\rho_i(n)}{\rho_j(n)} > (1-\delta) \log\frac{1}{\epsilon}, \forall j \neq i \Big\} < \\
&\hspace*{2.93in} T^*_i | \theta=i\Big)\\  
&= P\Big( \Big\{ \exists n, 0\le n < T^*_i \ \text{s.t.} \ \log\frac{\rho_i(n)}{\rho_j(n)} > \\
&\hspace*{1.83in} (1-\delta) \log\frac{1}{\epsilon}, \forall j \neq i \Big\} | \theta=i \Big )\\
&\stackrel{(a)}{\le} P\Big(\bigcup_{j \neq i} \Big\{ \exists n, 0\le n < T^*_i \ \text{s.t.} \ \log\frac{\rho_i(n)}{\rho_j(n)} - \mathbb{E}[\log\frac{\rho_i(n)}{\rho_j(n)}]> \\
&\hspace*{.72in}  (1-\delta) \log\frac{1}{\epsilon} - \max_{k\neq i} \log\frac{\rho_i}{\rho_k} - n R(i,\lambda^*_i)\Big\} | \theta=i\Big)\\
&\stackrel{(b)}{\le} P\Big(\bigcup_{j \neq i}\Big\{ \exists n, 0\le n < T^*_i \ \text{s.t.} \ \log\frac{\rho_i(n)}{\rho_j(n)} - \mathbb{E}[\log\frac{\rho_i(n)}{\rho_j(n)}]> \\
&\hspace*{2.75in} T^*_i \delta \Big\} | \theta=i \Big)\\
&\le \sum_{j \neq i} P \Big ( \max_{0\le n < T^*_i} \Big\{ \log\frac{\rho_i(n)}{\rho_j(n)} - \mathbb{E}[\log\frac{\rho_i(n)}{\rho_j(n)}] \Big\} >\\
&\hspace*{2.85in} T^*_i \delta  | \theta=i \Big )\\
&\stackrel{(c)}{\le} \frac{T^*_i (\log\xi)^2}{(T^*_i \delta)^2}\\
&= \frac{(\log\xi)^2}{T^*_i \delta^2},
\end{align*}
where $(a)$ follows from \eqref{Edrift}; 
$(b)$ follows from the definition of $T^*_i$ and the fact that $n < T^*_i$; and
$(c)$ follows from Fact~\ref{Kolmogorov}.

\end{IEEEproof}

\bibliographystyle{IEEEtran}
\bibliography{HypTest} 

\begin{thebibliography}{10}
\providecommand{\url}[1]{#1}
\csname url@samestyle\endcsname
\providecommand{\newblock}{\relax}
\providecommand{\bibinfo}[2]{#2}
\providecommand{\BIBentrySTDinterwordspacing}{\spaceskip=0pt\relax}
\providecommand{\BIBentryALTinterwordstretchfactor}{4}
\providecommand{\BIBentryALTinterwordspacing}{\spaceskip=\fontdimen2\font plus
\BIBentryALTinterwordstretchfactor\fontdimen3\font minus
  \fontdimen4\font\relax}
\providecommand{\BIBforeignlanguage}[2]{{%
\expandafter\ifx\csname l@#1\endcsname\relax
\typeout{** WARNING: IEEEtran.bst: No hyphenation pattern has been}%
\typeout{** loaded for the language `#1'. Using the pattern for}%
\typeout{** the default language instead.}%
\else
\language=\csname l@#1\endcsname
\fi
#2}}
\providecommand{\BIBdecl}{\relax}
\BIBdecl

\bibitem{Chernoff59}
H.~Chernoff, ``Sequential design of experiments,'' \emph{The Annals of
  Mathematical Statistics}, vol.~30, pp. 755--770, 1959.

\bibitem{Shenoy11}
P.~Shenoy and A.~J. Yu, ``Rational decision-making in inhibitory control,''
  \emph{Frontiers in Human Neuroscience}, vol.~5, no.~48, 2011.

\bibitem{Burnashev76}
M.~V. Burnashev, ``{Data transmission over a discrete channel with feedback.
  Random transmission time},'' \emph{Problemy Peredachi Informatsii}, vol.~12,
  no.~4, pp. 10--30, 1975.

\bibitem{ubli1}
G.~Thatte, U.~Mitra, and J.~Heidemann, ``Parametric methods for anomaly
  detection in aggregate traffic,'' \emph{IEEE/ACM Transactions on Networking},
  vol.~19, no.~2, pp. 512--525, April 2011.

\bibitem{ubli2}
G.~A. Hollinger, U.~Mitra, and G.~S. Sukhatme, ``Active classification: theory
  and application to underwater inspection,'' 2011, available on
  arXiv:1106.5829.

\bibitem{Nowak11IT}
R.~D. Nowak, ``The geometry of generalized binary search,'' \emph{IEEE
  Transactions on Information Theory}, vol.~57, no.~12, pp. 7893--7906,
  December 2011.

\bibitem{saligrama}
C.~L. Chan, P.~H. Che, S.~Jaggi, and V.~Saligrama, ``Non-adaptive probabilistic
  group testing with noisy measurements: Near-optimal bounds with efficient
  algorithms,'' in \emph{49th Annual Allerton Conference on Communication,
  Control, and Computing}, 2011, pp. 1832--1839.

\bibitem{Hero11}
A.~O. Hero and D.~Cochran, ``Sensor management: past, present, and future,''
  \emph{IEEE Sensors Journal}, vol.~11, no.~12, pp. 3064--3075, December 2011.

\bibitem{Nowak11}
M.~Malloy and R.~Nowak, ``Sequential analysis in high-dimensional multiple
  testing and sparse recovery,'' in \emph{IEEE International Symposium on
  Information Theory (ISIT)}, 2011, pp. 2661--2665.

\bibitem{Iwen09}
M.~A. Iwen, ``Group testing strategies for recovery of sparse signals in
  noise,'' in \emph{Forty-Third Asilomar Conference on Signals, Systems and
  Computers}, 2009, pp. 1561--1565.

\bibitem{Wald48}
A.~Wald and J.~Wolfowitz, ``Optimal character of the sequential probability
  ratio tests,'' \emph{The Annals of Mathematical Statistics}, vol.~19, no.~3,
  pp. 326--339, 1948.

\bibitem{Armitage50}
P.~Armitage, ``Sequential analysis with more than two alternative hypotheses,
  and its relation to discriminant function analysis,'' \emph{Journal of the
  Royal Statistical Society, Series B}, vol.~12, no.~1, pp. 137--144, 1950.

\bibitem{Lorden77}
G.~Lorden, ``Nearly-optimal sequential tests for finitely many parameter
  values,'' \emph{The Annals of Statistics}, vol.~5, no.~1, pp. 1--21, 1977.

\bibitem{Blahut74}
R.~E. Blahut, ``Hypothesis testing and information theory,'' \emph{IEEE
  Transactions on Information Theory}, vol.~20, no.~4, pp. 405--417, July 1974.

\bibitem{Tuncel05}
E.~Tuncel, ``On error exponents in hypothesis testing,'' \emph{IEEE
  Transactions on Information Theory}, vol.~51, no.~8, pp. 2945--2950, August
  2005.

\bibitem{Hayashi09}
M.~Hayashi, ``Discrimination of two channels by adaptive methods and its
  application to quantum system,'' \emph{IEEE Transactions on Information
  Theory}, vol.~55, no.~8, pp. 3807--3820, August 2009.

\bibitem{PolyanskiyITA2011}
Y.~Polyanskiy and S.~Verdu, ``Hypothesis testing with feedback,'' in
  \emph{Information Theory and Applications Workshop (ITA)}, 2011.

\bibitem{NitinawaratArxiv}
S.~Nitinawarat, G.~Atia, and V.~V. Veeravalli, ``Controlled sensing for
  multihypothesis testing,'' 2012, available on arXiv:1205.0858.

\bibitem{NitinawaratICASSP12}
------, ``Controlled sensing for hypothesis testing,'' in \emph{IEEE
  International Conference on Acoustics, Speech, and Signal Processing
  (ICASSP)}, March 2012.

\bibitem{HypJournal}
M.~Naghshvar and T.~Javidi, ``Active sequential hypothesis testing,'' 2012,
  available on arXiv:1203.4626.

\bibitem{Kumar}
P.~R. Kumar and P.~Varaiya, \emph{Stochastic systems: estimation,
  identification, and adaptive control}.\hskip 1em plus 0.5em minus 0.4em\relax
  Prentice-Hall, Inc., 1986.

\bibitem{Haroutunian07}
E.~A. Haroutunian, M.~E. Haroutunian, and A.~N. Harutyunyan, \emph{Reliability
  criteria in information theory and in statistical hypothesis testing}.\hskip
  1em plus 0.5em minus 0.4em\relax Hanover, MA, USA: Now Publishers Inc.,
  January 2007, vol.~4, no.~2.

\bibitem{ISIT11}
M.~Naghshvar and T.~Javidi, ``Performance bounds for active sequential
  hypothesis testing,'' in \emph{IEEE International Symposium on Information
  Theory (ISIT)}, 2011, pp. 2666--2670.

\bibitem{Allerton10}
------, ``Information utility in active sequential hypothesis testing,'' in
  \emph{Proceedings of the 48th Allerton conference on communication, control,
  and computing}, 2010.

\bibitem{CISS2012}
------, ``Active hypothesis testing: sequentiality and adaptivity gains,'' in
  \emph{Conference on Information Sciences and Systems (CISS)}, March 2012.

\bibitem{Shayevitz11ISIT}
O.~Shayevitz, ``{On R\'enyi measures and hypothesis testing},'' in \emph{IEEE
  International Symposium on Information Theory (ISIT)}, 2011, pp. 894--898.

\bibitem{Blackwell53}
D.~Blackwell, ``Equivalent comparisons of experiments,'' \emph{The Annals of
  Mathematical Statistics}, vol.~24, pp. 265--272, 1953.

\bibitem{Sakaguchi64}
M.~Sakaguchi, ``Information theory and decision making,'' 1964, unpublished
  lecture notes, Statistic Deptartment, The George Washington University,
  Washington, D.C.

\bibitem{Csiszar04}
I.~Csisz\'{a}r and P.~Shields, \emph{Information theory and statistics: a
  tutorial}.\hskip 1em plus 0.5em minus 0.4em\relax Now Publishers Inc.,
  December 2004, vol.~1, no.~4.

\bibitem{Asilomar10}
M.~Naghshvar and T.~Javidi, ``Variable-length coding with noiseless feedback
  and finite messages,'' in \emph{Conference Record of the Forty Fourth
  Asilomar Conference on Signals, Systems and Computers}, 2010, pp. 317--321.

\bibitem{Berlekamp64}
E.~R. Berlekamp, ``Block coding with noiseless feedback,'' 1964, {PhD Thesis,
  MIT, Cambridge, MA}.

\bibitem{NakibogluPhD}
B.~Nakiboglu, ``Exponential bounds on error probability with feedback,'' 2011,
  {PhD Thesis, MIT, Cambridge, MA}.

\bibitem{Burnashev80}
M.~V. Burnashev, ``Sequential discrimination of hypotheses with control of
  observations,'' \emph{Math. USSR Izvestija}, vol.~15, no.~3, pp. 419--440,
  1980.

\bibitem{Bentkus08}
V.~Bentkus, ``An extension of the hoeffding inequality to unbounded random
  variables,'' \emph{Lithuanian Mathematical Journal}, vol.~48, no.~2, pp.
  137--157, 2008.

\bibitem{Dragalin99}
V.~P. Dragalin, A.~G. Tartakovsky, and V.~V. Veeravalli, ``Multihypothesis
  sequential probability ratio tests. {I}. {A}symptotic optimality,''
  \emph{IEEE Transactions on Information Theory}, vol.~45, no.~7, pp.
  2448--2461, November 1999.

\bibitem{Sundaresan12}
N.~Vaidhiyan, S.~Arun, and R.~Sundaresan, ``Active sequential hypothesis
  testing with application to a visual search problem,'' in \emph{IEEE
  International Symposium on Information Theory (ISIT)}, 2012, pp. 2201--2205.

\bibitem{Billingsley}
P.~Billingsley, \emph{Probability and measure}.\hskip 1em plus 0.5em minus
  0.4em\relax John Wiley \& Sons, Inc., 1995.

\bibitem{McDiarmid89}
C.~McDiarmid, ``On the method of bounded differences,'' \emph{Surveys in
  Combinatorics, London Mathematical Society Lecture Note Series 141, Cambridge
  University Press}, pp. 148--188, 1989.

\end{thebibliography}

\end{document}